\documentclass[11pt,authoryear]{elsarticle}
\usepackage[margin=2.3cm, includehead]{geometry}

\usepackage{currfile}
\usepackage{mathrsfs}
\usepackage[mathscr]{eucal}
\DeclareMathAlphabet{\mathpzc}{OT1}{pzc}{m}{it}
\usepackage{natbib}
\usepackage{makeidx}
\usepackage{multirow}
\usepackage{multicol}
\usepackage{booktabs}
\usepackage{float}
\usepackage{hhline}
\usepackage[dvipsnames,svgnames,table]{xcolor}
\usepackage{graphicx}
\usepackage{epstopdf}
\usepackage{ulem}
\usepackage{amsmath}
\usepackage{amsfonts}
\usepackage{url}
\usepackage{amssymb}
\usepackage{graphics}
\usepackage{epsfig}
\usepackage{verbatim}
\usepackage{amsthm}
\usepackage{geometry}
\usepackage{framed}    
\usepackage{scalefnt}
\usepackage{alltt}
\usepackage{hyperref}
\hypersetup{
    bookmarks=true,         
    unicode=false,          
    pdftoolbar=true,        
    pdfmenubar=true,        
    pdffitwindow=true,      
    pdftitle={My title},    
    pdfauthor={Author},     
    pdfsubject={Subject},   
    pdfnewwindow=true,      
    pdfkeywords={keywords}, 
    colorlinks=true,        
    linkcolor=blue,         
    citecolor=blue,        
    filecolor=magenta,      
    urlcolor=cyan           
}

\def\i.i.d.{\buildrel {\rm i.i.d.} \over \sim}

\def\cw#1 { \overset{\mathbb{P}}{\underset{#1}{\longrightarrow}} }
\def\Real{\mathbb{R}}
\def\Natu0{\mathbb{N}_0}
\def\P#1{{\mathrm{P}}\left(#1\right)}

\def\E#1{{\mathrm E}\left[#1\right]}

\def\Var#1{{\mathrm Var}\left(#1\right)}

\def \rcov#1#2 {{\rm cov}_{#1}\left( #2\right)}

\newtheorem{lemma}{Lemma}
\newtheorem{theorem}{Theorem}

\newtheorem{corollary}{Corollary}
\newtheorem{remark}{Remark}

\newtheorem{model}{Model}

\begin{document}
	\begin{titlepage}
		
		\thispagestyle{empty}
\title{
A simple Bayesian state-space model for the collective risk model.(\currfilename.)}

%
%
%
%
%

\author{Jae Youn Ahn\corref{bbb}\fnref{thirdfoot}}
\author{Himchan Jeong \fnref{firstfoot}}
\author{Yang Lu\corref{bbb}\fnref{secondfoot}}

\cortext[bbb]{Corresponding Authors}
\fntext[thirdfoot]{Department of Statistics, Ewha Womans University, Seoul, Republic of Korea. Email: \url{jaeyahn@ewha.ac.kr}}
\fntext[firstfoot]{Department of Statistics and Actuarial Science, Simon Fraser University, BC, Canada. Email: \url{himchan_jeong@sfu.ca}}
\fntext[secondfoot]{Department of Mathematics and Statistics, Concordia University, Montreal, QC, Canada. Email: \url{yang.lu@concordia.ca}.}



\begin{abstract}
The collective risk model (CRM) for frequency and severity is an important tool for retail insurance ratemaking,  macro-level catastrophic risk forecasting, as well as operational risk in banking regulation. This model, which is initially designed for cross-sectional data, has recently been adapted to a longitudinal context 
to conduct both \textit{a priori} and \textit{a posteriori} ratemaking, through the introduction of random effects. However, so far, the random effect(s) is usually assumed static due to computational concerns, leading to predictive premium that omit the seniority of the claims.
In this paper, we propose a new CRM model with bivariate dynamic random effect process. The model is based on Bayesian state-space models. It is associated with the simple predictive mean {and closed form expression for the likelihood function,} while also allowing for the dependence between the frequency and severity components. Real data application to auto insurance is proposed to show the performance of our method.

%

\end{abstract}


\end{titlepage}


\maketitle




\vfill

\pagebreak

\section{Introduction}

The frequency-severity collective risk model (CRM) is an important tool for retail insurance ratemaking,  macro-level catastrophic risk forecasting, as well as operational risk in banking regulation.  
Among the early contributions, \citet{frees2014predictive} is the first to account for dependence between the frequency and severity components, by using the frequency as an explanatory variable in the regression equation of the severity variable, and \citet{garrido2016generalized} extended this model under the framework of exponential dispersion family.  On the other hand, \cite{czado2012mixed, shi2018pair,cossette2019collective,lee2019dependent, yang2020nonparametric,oh2021multi} use copulas to capture the dependence between the frequency and severity components.  However, so far most of the aforementioned (regression, or copula-based) models are designed for cross-sectional data only\footnote{To our knowledge, only the model of \cite{lee2019dependent} can be applied to panel data. This copula approach, however, will not be investigated in the present paper for several reasons. First, the use of copulas when count data are involved is not without debate due to potential identification failures [see \cite{genest2007primer}]. Recently, some advances have been made in the case when all the marginals are count-valued [\cite{yang2020nonparametric}]. This result, however, does not apply to mixed, frequency-severity data, which has the specificity that the severity is equal to zero, if and only if the frequency is zero.  Secondly, the copula and the random effect approaches are conceptually very different and their comparison is difficult. As \cite{frees2005credibility} put it, ``the frequentist (copula-based) perspective avoids assumption concerning the prior distribution of the latent variables. This may be seen as an advantage or disadvantage, depending on the situation."
}. Recently,  random effects have been introduced in longitudinal, or panel CRM's \citep{jeong2021generalized, lu2019flexible, jeong2020predictive, cheung2021bayesian, denuit2021wishart, oh2021predictive}.  The random effect creates serial correlation between observation of different periods, as well as (serial and cross-sectional) dependence between the frequency and severity processes.  Moreover, many of these models
 come with closed form prediction formulas.

Nevertheless, the random effects in these models are all static, or time-invariant. In  a longitudinal context, this has the downside of not distinguishing the seniority of the claims. This might be counterintuitive for many applications such as auto insurance, in which the most recent claim experiences are believed to be of better predictive power of future claims. 
One natural solution to this issue is to introduce dynamic random effect(s). This idea has already been implemented in frequency-only models \citep{pinquet2001allowance, lu2018dynamic}, but has yet to be extended to CRM's. One of the possible reasons is the computational burden.  Indeed, according to \cite{cox1981statistical}'s classification of time series,  existing dynamic random effects models employed in the frequency literature are \textit{parameter-driven}, in the sense that the random effect process has its own, exogenous dynamics.  These models are usually associated with complicated posterior premium, even in the simpler, frequency models. In the CRM framework, a second sequence of random effects is required for the severity part. This makes the computation task even more formidable, especially if the dependence between the frequency and severity components is also to be accounted for. Let us also mention the credibility approach. \cite{pinquet2001allowance} show that dynamic credibility models can be applied to get linear prediction formula for \textit{frequency-only}, or \textit{severity-only} models that account for seniority.  However, to our knowledge, no \textit{dynamic} credibility models have been proposed to account for \textit{both} the frequency and severity part in a non-trivial way, i.e.  allowing for their interdependence.


This paper proposes a dynamic, random effect based CRM model that $i)$ accounts for the cross-sectional dependence between the frequency and severity parts; $ii)$ provides the closed form solution for the a posteriori ratemaking; {$iii)$ provides the closed form expression of the likelihood function. } Our model is based on a time series literature on Bayesian exponential family state space models, which is pioneered by \cite{bayesianforecasting} in the Gaussian case, extended and popularized in a series of papers such as \cite{smith1979generalization, smith1986non} (henceforth SM) and \cite{harvey1989time} (henceforth HF) in the context of nonnegative continuous (resp. count) valued time series.  These models are also state-space (or latent factor, dynamic random effect) based, but are \textit{observation-driven} instead of being \textit{parameter-driven}, in the sense that the dynamics of the random effect process depends also on the observed frequency. SM and HF show that these models have simple forecasting formulas, when appropriate updating rules involving conjugate priors are used for the random effect process, hence the term \textit{Bayesian} state-space model.  In other words, the difference between the HF and SM's approach and the aforementioned parameter-driven state-space models is similar to the difference between GARCH and stochastic volatility models for asset return data. Both models are popular in the finance literature, but the former is much more simpler due to its closed form likelihood function. Recently, it has been shown by \citet{koopman2016predicting} that besides their computational advantages, observation-driven models are also quite competitive in terms of forecasting accuracy for various types of data, even in the presence of mis-specification error.  Such Bayesian state space models have been successfully used in many financial/economic applications, including to univariate asset returns [\cite{shephard1994local}], as well as to high-dimensional macroeconomic data [\cite{uhlig1997bayesian}], and they can be easily extended to other distributions belonging to the exponential family \citep{grunwald1993prediction, vidoni1999exponential}.

This time series tool, however, remains relatively under-explored in the actuarial literature and to our knowledge, only the count time series model of HF has been applied to claim frequency processes \citep{bolance2007greatest, abdallah2016sarmanov,boucher2018claim}.  Thus, the first aim of the paper is to introduce the SM model to the insurance literature. We do so by first reminding the HF model, and by emphasizing on the similarity of the underlying  ideas of these two models.
By doing so, we also provided a generalized version of the SM model to be applicable for the ratemaking in insurance setting.
This generalized version of the SM model is applicable to any positive valued time series and hence  might be of interest beyond the CRM framework.  Then we turn our eyes to the application of HF and SM models to CRM.
Because these two models concern only univariate processes with counts and positive values, respectively, the second aim of the paper is to combine these two frameworks into a dynamic, CRM, so that at each period, a count, frequency variable is observed with a nonnegative, severity variable, with the severity variable taking value zero if and only if the count variable is zero.  This kind of combination has, to our knowledge, never been considered before in the time series or actuarial literature. Moreover, this extension is not straightforward, because on the contrary to the univariate cases in the HF and SM models,  there is no obvious bivariate conjugate priors for the CRM, if we were to allow for dependence between the frequency and severity processes.   A similar difficulty has already been reported by \cite{abdallah2016sarmanov}, who investigate the possibility of extending the HF model to bivariate frequency processes. They propose to use a Sarmanov distribution to couple the two random effect components with gamma marginal distributions, but acknowledge that the conjugacy is lost in this bivariate case and they have to replace the posterior distribution by a more tractable \textit{approximation}.  They show that the approximation error is not large, but nonetheless non-negligible.  The solution we propose here is inspired by the cross-sectional model of \cite{garrido2016generalized}. On the one hand, we require the \textit{predictive distribution} of the bivariate random effect process to have independent components. This allows us to deduce prediction formula that are \textit{exact} and in closed form; on the other hand, this previous assumption,  however, is weaker than outright independence between the two component processes of the random effect. In particular, in our framework these two components are dependent such that the predictive distribution of the severity depends on the number of counts of the same period, in a similar way as in \cite{garrido2016generalized}.  Our final contribution is that we also discuss variants of our modeling approach, by explicitly singling out periods with zero claims and allowing for potentially different updating formulas for these periods. This leads to a \textit{three-part} model, in which not only frequency and severity are separately, but also claim and non-claim periods are distinguished.

The paper is organized as follows.  Section 2 reviews some basic results related to gamma distribution, as well as the univariate HF model. Section 3 provides a generalized version of the univariate SM model.
Section 4 introduces the CRM model and works out the prediction formulas. Section 5 discusses the link of the model with other two-part models in the econometric literature, and introduces a three-part variant that allows the updating rule of the random effect to be different depending on whether the claim frequency is zero. Section 6 proposes an empirical data illustration. Section 7 concludes.

\section{Review of the HF model}
This section provides a quick reminder of the general setting, and the preliminary results concerning HF.

\subsection{Notation and definition}\label{sec.2.1}

 For a given individual $i$, let us denote by
\begin{itemize}
\item $y_{t}^{[1]}\in\mathbb{N}$ the claim count at time $t=1,2,...$, i.e., the \textit{frequency}, and $\left\{ \mathcal{F}_t^{[1]}\right\} _{t=0}^{\infty}$
denotes the natural filtration generated by $y_{t}^{[1]}$;
\item $y_{t}^{[2]}\in\Real_0^{+}$ the aggregate claim at time $t$.
\end{itemize}
This framework, with two response variables per period, is called the frequency-severity two-part model,
and we use $\left\{ \mathcal{F}_t\right\}_{t=0}^{\infty}$ to denote the natural filtration
generated by $\left(y_{t}^{[1]}, y_{t}^{[2]}\right)$.

Throughout the paper, we denote by
\begin{itemize}
  \item ${\rm Gamma}(\alpha, \beta)$: the gamma distribution with shape parameter $\alpha$  and
rate parameters $\beta$. In other words, the mean and variance of this distribution are $\alpha/\beta$ and $\alpha/\beta^2$, respectively.
 {As a convention, we use $Y\sim{\rm Gamma}(0, \beta)$ for $\P{X=0}=1$.}
 \item ${\rm IG}(\alpha, \beta)$: the inverse gamma distribution with parameter $(\alpha, \beta)$, which is the distribution of $\frac{1}{Y}$, where $Y$ follows ${\rm Gamma}(\alpha, {\beta})$ distribution. In particular, if $\alpha>1$, then its mean exists and is equal to:
\[
\frac{\beta}{\alpha-1},
\]
and if $\alpha>2$, then the inverse gamma distribution has a finite variance:
\[
\frac{\beta^2}{(\alpha-1)^2 (\alpha-2)}.
\]
 \item ${\rm Beta}(\alpha, \beta)$: the beta distribution on $(0,1)$
%
whose mean and variance are
\[
\frac{\alpha}{\alpha+\beta}\quad\hbox{and}\quad
\frac{\alpha\beta}{(\alpha+\beta)^2 (\alpha+\beta+1)},
\]
respectively.
\item ${\rm NB}(\lambda, \Gamma)$: the negative binomial distribution
whose mean and variance are given by
\[
\lambda\quad\hbox{and}\quad \lambda + \lambda^2/\Gamma,
\]
respectively.
Note that, for $Y\sim{\rm NB}(\lambda, \Gamma)$,
the moment generating function of $Y$ is given by
\[
\E{\exp\left(zY\right)}=\left(\frac{\Gamma}{\Gamma+\lambda-\lambda\exp(z)} \right)^\Gamma, \quad |z|<1+\frac{\Gamma}{\lambda}
\]
and the probability mass function at $Y=y$ is given by
\[
f_{\rm NB}(y;\lambda, \Gamma)=
\frac{\Gamma(y+\Gamma)}{y!\Gamma(\Gamma)}\left( \frac{\Gamma}{\Gamma+\lambda}\right)^\Gamma
\left( \frac{\lambda}{\Gamma+\lambda}\right)^y.
\]

\item ${\rm GB2}(a,b,p,q)$: the GB2 distribution having the following density function at $Y=y$
\[
f_{\rm GB2}\left(y; a,b,p,q \right) =
\frac{|a|y^{ap-1}}
{b^{ap}B(p,q)\left(1+(y/b)^a \right)^{p+q}},  \quad y> 0
\]
{As a convention, we use $X\sim{\rm GB2}(a,b,0,0)$ for $\P{Y=0}=1$.}
We are particularly interested in the case when $a=1$,  which is also called generalized Pareto distribution as defined in the appendix of \citet{klugman2012loss} and widely used in the actuarial literature such as \citet{jeong2020bayesian}. In this case, the mean and variance of ${\rm GB2}(1,b,p,q)$ are conveniently given by
\[
\frac{bp}{q-1}\quad\hbox{and}\quad \frac{b^2p}{q-1}\left( \frac{p+q-1}{(q-2)(q-1)}\right).
\]
\end{itemize}

Finally, we denote the probability density/mass function at $y$ by
\[
f_{\rm Gamma}(y;\alpha, \beta), \quad f_{\rm IG}(y;\alpha, \beta), \quad
 \quad\hbox{and}\quad f_{\rm NB}(y; \lambda, \Gamma),
\]
and the corresponding cumulative distribution functions by
\[
F_{\rm Gamma}(y;\alpha, \beta), \quad F_{\rm IG}(y;\alpha, \beta), \quad
 \quad\hbox{and}\quad F_{\rm NB}(y; \lambda, \Gamma),
\]
respectively.
We also find it convenient to use $\prod$ to represent the independent copula. For example,
\[
(X,Y)\sim{\rm \prod}\left( F_{\rm Gamma}\left(\cdot;\alpha_1, \beta_1\right), F_{\rm Gamma}\left(\cdot;\alpha_2, \beta_2 \right)\right)
\]
 to denote two independent gamma distribution with
\[
X\sim {\rm Gamma}\left( \alpha_1, \beta_1\right) \quad\hbox{and}\quad
Y\sim {\rm Gamma}\left( \alpha_2, \beta_2\right).
\]

\subsection{Review of the HF model for counts}\label{sec.2.2}

First, the following lemma, called the characterization of the gamma distribution, is the key idea of the HF dynamic random effect models.
\begin{lemma}[\cite{lukacs1955characterization}]\label{lem.1}
Given independent $\theta\sim {\rm Gamma}(\alpha, \beta)$, $B\sim {\rm Beta}(q\alpha, (1-q)\alpha)$, and $q>0$ variables, we have
 \[
  \theta B \sim {\rm Gamma}(q\alpha, \beta).
  \]
\end{lemma}
It is easy to check that
\begin{equation}
\label{randomwalk}
\E{  \frac{1}{q} B \theta \mid \theta}=\theta.
\end{equation}
In other words, this re-scaling allows to construct recursively a positive  martingale. Moreover, a simple application of Lemma 1 shows that
  \[
  \frac{1}{q} B \theta \sim {\rm Gamma}(q\alpha, q\beta).
  \]
  In other words, the positive martingale (or random walk without drift) we construct has gamma marginal distribution for each fixed time $t$.  This kind of martingale type models is also called steady-state \citep{bayesianforecasting}, or local-level \cite{shephard1994local} in the Bayesian time series literature. Their advantages are that first, the mean of a process constructed using this technique is time-invariant. Second, in many of the aforementioned models including the HF model, this multiplicative martingale model leads to a simple forecasting formula that is an exponentially weighted moving average (EWMA).

Now we are ready to present HF model proposed in \cite{harvey1989time}.

\begin{model}[\cite{harvey1989time}]
\label{mod.0}
For constants $q\in(0,1)$, consider the stochastic process
\[
\left(y_t, \theta_t \right)_{t\in\mathbb{N}}
\]
where the univariate observations $y_t$ are counts and state variables (or random effects) $\theta_t$ are positive, with the following joint dynamics
\begin{enumerate}
  \item[i. ] (Initial condition for the state variable) At the initial date, the random effect is gamma distributed:
  \[
  \theta_0\sim {\rm Gamma}\left(\alpha_0, \beta_0\right)
  \]
	with positive parameters $\alpha_0, \beta_0$.

 \item[ii.] (Transition equation at time $t$): Assume that at a certain time $t$,  the filtering distribution $ \theta_{t-1}|y_1, \cdots, y_{t-1} $ is gamma 
 \begin{equation}\label{eq.3400}
 \theta_{t-1}|y_1, \cdots, y_{t-1} \sim {\rm Gamma}\left(   \alpha_{t-1}, \beta_{t-1}\right)
 \end{equation}
with positive parameters $\alpha_{t-1}$, $\beta_{t-1}$, and let
$B_t$ be conditionally independent with $\theta_{t-1}$ given $y_1, \cdots, y_{t-1}$, with beta marginals:
 \[
 B_t|y_1, \cdots, y_{t-1}\sim {\rm Beta}\left( q \alpha_{t-1}, (1-q) \alpha_{t-1}\right).
 \]
 Then, the new state variable $\theta_t$ is defined by
 \[
\theta_t:= \frac{\theta_{t-1}B_t}{q}.
\]
\item[iii.] (Observed variable at time $t$) For a given positive constant $\lambda>0$
which might depend on individual characteristics, the observation at time $t$, $y_t$ is drawn from the following conditional distribution:
    \[
    y_t | y_1, \cdots, y_{t-1}, \theta_1, \cdots, \theta_{t} \sim {\rm Pois}\left( \lambda_t\theta_t\right).
    \]
\end{enumerate}
\end{model}

Because the transition equation \eqref{eq.3400} at time $t$ in Model \ref{mod.0} assumes the distribution of the filtering distribution at time $t-1$ to be gamma:
 \begin{equation*}
 \theta_{t}|y_1, \cdots, y_t \sim {\rm Gamma}\left( \alpha_{t}, \beta_{t}\right), \qquad \forall t
 \end{equation*}
we need to check whether this assumption is compatible with the Bayes formula, which is usually used to compute the filtering distribution. The following result, due to \cite{harvey1989time}, confirms this is indeed the case, and is based on the Poisson-gamma conjugacy.  We also obtain, as a by-product, the predictive distribution of the random effect, which is also gamma.

 \begin{lemma} [Recursion for forecasting and filtering]\label{lem.a1}
 For a fixed time $t$, 
we have the following results under Model \ref{mod.0}
 \begin{enumerate}
 \item[i.] One-step-ahead forecasting of the random effects at time $t$ is given by
 \[
 \theta_{t}|y_1, \cdots, y_{t-1} \sim {\rm Gamma}\left(  q\alpha_{t-1}, q\beta_{t-1}\right).
 \]
 \item[ii.] The one-step-ahead forecasting density for the observations is given by
\[
y_t\,|\, y_1, \cdots, y_{t-1} \sim {\rm NB}\left( \lambda_t\frac{\alpha_{t-1}}{\beta_{t-1}}, q\alpha_{t-1}\right).
\]
   \item[iii.] The filtering distribution at time $t$ is
    \[
 \theta_{t}|y_1, \cdots, y_t \sim {\rm Gamma}\left( \alpha_{t}, \beta_{t}\right)
 \]
 with
   \[
   \begin{cases}
     \alpha_{t}:=q\alpha_{t-1}+y_t>0;\\
     \beta_{t}:=q\beta_{t-1}+\lambda_t>0;\\
   \end{cases}
   \]
 \end{enumerate}
 \end{lemma}

We note that the closed form expression of the likelihood function can be obtained from  the one-step-ahead forecasting density for the observations in Lemma \ref{lem.a1}.
An immediate consequence of Lemma \ref{lem.a1} is the following relationships concerning the conditional mean and variance
      \begin{equation}\label{up.1}
      \E{\theta_{t} \,|\, y_1, \cdots, y_{t-1}}=
      \E{\theta_{t-1} \,|\,  y_1, \cdots, y_{t-1}}.
      \end{equation}
and
      \begin{equation}\label{up.2}
      \Var{\theta_{t} \,|\, y_1, \cdots, y_{t-1}}=\frac{1}{q}
      \Var{\theta_{t-1} \,|\, y_1, \cdots, y_{t-1}}.
      \end{equation}
which intuitively explains the dynamics of the random effects in Model \ref{mod.0}: it preserves the mean while inflating the variance in one-step-ahead forecasting of the random effects. Finally, we provide the one-step-ahead forecasting for Model \ref{mod.0}.

\begin{lemma}\label{cor.a1}
    Under Model \ref{mod.0},  the one-step-ahead forecasting of the frequency at time $\tau$ is given by
      \begin{equation}\label{eq.4}
      \begin{aligned}
    \E{y_\tau \,|\, y_1, \cdots, y_{\tau-1} }&=
    \lambda_\tau
    \frac{\sum\limits_{t=1}^{\tau-1} q^{\tau-1-t}y_t +
    q^{\tau-1} \alpha_0}
    {\sum\limits_{t=1}^{\tau-1}q^{\tau-1-t}\lambda_{t} +
    q^{\tau-1} \beta_0}\\
&=\lambda_\tau
\left[
b_0+\sum\limits_{t=1}^{\tau-1}b_t\frac{y_t}{\lambda_t}
\right]
    \end{aligned}
\end{equation}
where
\[
b_0:=\frac{q^{\tau-1}\alpha_0}{\sum\limits_{t=1}^{\tau-1}q^{\tau-1-t}\lambda_{t} +
    q^{\tau-1} \beta_0}
\quad\hbox{and}\quad
b_t:=\frac{q^{\tau-1-t}\lambda_{t}}{\sum\limits_{t=1}^{\tau-1}q^{\tau-1-t}\lambda_{t} +
    q^{\tau-1} \beta_0}, \quad t=1, \cdots, \tau-1.
\]

\end{lemma}

\section{Introduction of a generalized SM model for continuous variables}
The key idea of the above HF model  is in Lemma \ref{lem.1} and the Poisson-gamma conjugacy, which allow us to stay within the gamma family when we alternate between the filtering and predictive distribution of the random effect process. The SM model has a similar idea, but relies instead on the inverse gamma-gamma conjugacy.
In this section, we extend the SM model in \cite{smith1986non} where $q_t$ and $q_t^*$ are assumed to be fixed constants.
The following model is a generalized version of the SM model in \cite{smith1986non}.
We note that, while we allow  $q_t$ and $q_t^*$ to be varying, the SM model in \cite{smith1986non} assume $q_t$ and $q_t^*$ to be fixed constants.

\begin{model}\label{mod-1}
	Consider the stochastic process
	\[
	\left(y_t, \theta_t \right)_{t\in\mathbb{N}}
	\]
	where the observations $y_t$ are positive, real and state variables (or random effects) $\theta_t$ are positive, with the following joint dynamics
	\begin{enumerate}
		\item[i. ] (Initial condition for the state variable) At $t=0$, the random effect is inverse gamma distributed:
		\[
		\theta_0\sim {\rm IG}\left(\alpha_0, \beta_0\right)
		\]
		with parameters $\alpha_0>1, \beta_0>0$.
		
		\item[ii.] (Transition equation at time $t$): Assume that at a certain time $t$,  the filtering distribution $ \theta_{t-1}|y_1, \cdots, y_{t-1} $ is gamma 
		\begin{equation}\label{eq.3401}
			\theta_{t-1}|y_1, \cdots, y_{t-1} \sim {\rm IG}\left( \alpha_{t-1}, \beta_{t-1}\right)
		\end{equation}
		with parameters $\alpha_{t-1}>1$, $\beta_{t-1}>0$, and let
		$B_t$ be conditionally independent with $\theta_{t-1}$ given $y_1, \cdots, y_{t-1}$, with beta marginal:
		\[
		B_t|y_1, \cdots, y_{t-1}\sim {\rm Beta}\left(  q_t \alpha_{t-1}, (1-q_t) \alpha_{t-1}\right)
		\]
		{where $q_t\in(0,1)$ is a function of $\alpha_{t-1}$.}
	
		Then, the new state variable $\theta_t$ is defined by
		\begin{equation}\label{eq.2}
			\theta_t:= \frac{\theta_{t-1}q_t^*}{B_t},
		\end{equation}
		{where $q_t^*>0$ is another function of $\alpha_{t-1}$.}

		\item[iii.] (Observed variable at time $t$) For given positive constants $\lambda$ and $\psi$ (which might depend on individual characteristics), the observation at time $t$, $y_t$, is drawn from the following conditional distribution:
		\begin{equation*}
			y_t | y_1, \cdots, y_{t-1}, \theta_1, \cdots, \theta_{t} \sim {\rm Gamma}\left( \frac{1}{\psi}, \frac{1}{\theta_t\lambda_t\psi}\right).
		\end{equation*}
		
	\end{enumerate}
\end{model}

Similar as in Model \ref{mod.0}, the transition equations \eqref{eq.3401} at time $t$  requires the filtering distribution from time $t-1$ to be inverse gamma.  Lemma 5 below is the analog of Lemma 2.  It confirms that this is indeed the case due to the inverse gamma-gamma conjugacy. We also obtain, as a by-product, the predictive distribution of the random effect, which is also inverse gamma.

\begin{lemma} [Recursion for forecasting and filtering]\label{lem.a2}
	Under Model \ref{mod-1}
	we have
	\begin{enumerate}[$i.$]
		\item One-step-ahead forecasting of the random effects at each time $t$ is given by
		\[
		\theta_{t}|y_1, \cdots, y_{t-1} \sim {\rm IG}\left(  q_t\alpha_{t-1}, q_t^*\beta_{t-1}\right).
		\]
		\item The one-step-ahead forecasting density for the observations is given by
		\[
		y_t\,|\, y_1, \cdots, y_{t-1} \sim {\rm GB2}\left(
		1\,,\,
		{q_t^*}\beta_{t-1}\lambda_{t}\psi \,,\,
		\frac{1}{\psi}\,,\, 
		{q_t\alpha_{t-1}}
		\right)
		\]
		\item The filtering distribution at time $t$ is given as
		\[
		\theta_{t}|y_1, \cdots, y_t \sim {\rm IG}\left( \alpha_{t}, \beta_{t}\right)
		\]
		with
		\begin{equation}
			\label{parameterupdating}
		\begin{cases}
			\alpha_{t}:={q_t\alpha_{t-1}}+\frac{1}{\psi};\\
			\beta_{t}:={q_t^*\beta_{t-1}}+\frac{y_t}{\lambda_t\psi}.
		\end{cases}
		\end{equation}
	\end{enumerate}
\end{lemma}

We finally note that the closed form expression of the likelihood function can be obtained from  the one-step-ahead forecasting density for the observations in Lemma \ref{lem.a2}.

\subsection{Specification of $q_t$, $q_t^*$}
While the SM model in \cite{smith1986non} assume constant $q_t$ and $q_t^*$, for insurance applications,
 some other specifications might lead to more desirable properties. Indeed, first, inspired by the above HF model, we would like the process $(\theta_t)$ to have the same kind of martingale dynamics so that its mean does not explode when time increases. Second, because the one-step-ahead predictive, GB2 distribution only has finite mean and variance when $q_{t}\alpha_{t-1}$ is larger than 2, we would like this property to hold at any time $t$.

To this end, we first rewrite Lemma 1 in an equivalent, inverse gamma version. That is, given $q, \tilde{q}_1  \in(0,1)$,  $\alpha>1$, and $\tilde{q}_2, \beta>0$, random variables $\theta\sim {\rm IG}(\alpha, \beta)$ and $B\sim {\rm Beta}(\tilde{q}_1\alpha, (1-\tilde{q}_1)\alpha)$ are independent, then the rescaled product
\begin{equation*}
\theta^*:=\frac{\theta \tilde{q}_2}{B},
\end{equation*}
is ${\rm IG}(\tilde{q}_1 \alpha, \tilde{q}_2\beta)$ distributed with mean and variance
\begin{equation}
\E{\theta^*}=\frac{\tilde{q}_2 \beta}{\tilde{q}_1 \alpha-1}, \qquad
Var(\theta^*)=\frac{\tilde{q}_2^2\beta^2}{(\tilde{q}_1\alpha-1)^2(\tilde{q}_1\alpha-2)},
\end{equation}
 provided that $\tilde{q}_1\alpha>2.$


The following Lemma is a direct application of Lemma 1, by using slightly different scaling parameters in order to have a martingale with inverse gamma marginal. 

\begin{lemma}\label{lem.5}
For $q, \tilde{q}_1  \in(0,1)$,  $\alpha>1$, and $\tilde{q}_2, \beta>0$, assume that random variables, $\theta\sim {\rm IG}(\alpha, \beta)$ and $B\sim {\rm Beta}(\tilde{q}_1\alpha, (1-\tilde{q}_1)\alpha)$ are independent, and define the inverse gamma variable $\theta^*$ through:
\begin{equation*}
\theta^*:=\frac{\theta \tilde{q}_2}{B}.
\end{equation*}
 Then, $\E{\theta}$ is finite and equal to
\begin{equation}
	\label{equalmean}
	\E{\theta^*}=\E{\theta}=\frac{\beta}{\alpha-1}.
\end{equation}
if and only if \begin{equation}
	\label{relationq1q2}
	\tilde{q}_2 (\alpha-1)=\tilde{q}_1\alpha-1>0.
	\end{equation}
Furthermore, if $\alpha>2$ so that $Var(\theta)$ is finite, then $\theta^*$ satisfies simultaneously equation \eqref{equalmean} and
\begin{equation}
	\label{largervariance}
	\Var{\theta^*}=\frac{1}{q}\Var{\theta}< \infty, \quad q\in(0,1)
\end{equation}
if and only if
   \begin{equation}
   	\label{onlyq1q2}
      \tilde{q}_1=\frac{q(\alpha-2)+2}{\alpha}  \quad\hbox{and}\quad       \tilde{q}_2=\frac{q(\alpha-2)+1}{\alpha-1}.
      \end{equation}

\end{lemma}
Note, that $\tilde{q}_1, \tilde{q}_2$ defined in \eqref{onlyq1q2} satisfy:
      \begin{equation}\label{eq.b1}
	\frac{2}{\alpha}<\tilde{q}_1<1\quad\hbox{and}\quad 0<\tilde{q}_2<1.
\end{equation}


\begin{proof}


It is straightforward to check that the only solution to the above equations \eqref{equalmean} and \eqref{largervariance} is given by equation \eqref{onlyq1q2}.

%
%
%
%
%
%
\end{proof}

This lemma suggests the following specification of $q_t$ and $q_{t}^*$ in the SM model, which will be assumed throughout the rest of this section:
\begin{equation}
	\label{qt2}
	q_t:=\frac{q(\alpha_{t-1}-2)+2}{\alpha_{t-1}}\quad\hbox{and}\quad
	q_t^*:=\frac{q\left( \alpha_{t-1}-2\right)+1}{\alpha_{t-1}-1}, \quad q\in(0,1)
\end{equation}
Hence, under the model assumption in Model \ref{mod-1}, specification in \eqref{qt2} coupled with Lemma \ref{lem.a2} implies
\begin{equation}\label{eq.c1}
	\E{\theta_{t} \,|\, y_1, \cdots, y_{t-1}}=
	\E{\theta_{t-1} \,|\,  y_1, \cdots, y_{t-1}}
\end{equation}
and
\begin{equation*}
\E{y_t\,|\, y_1, \cdots, y_{t-1} }=
\lambda_t\frac{\beta_{t-1}}{\alpha_{t-1}-1}.
\end{equation*}
Moreover,  if $\alpha_{0}>2$, then using
$
q_t\alpha_{t-1}=q\left( \alpha_{t-1}-2\right)+2>2
$ from Lemma \ref{lem.5},
we have the following equality regarding the predictive variance of $\theta_t$:
      \begin{equation}\label{eq.c2}
	\Var{\theta_{t} \,|\, y_1, \cdots, y_{t-1}}=\frac{1}{q}
	\Var{\theta_{t-1} \,|\, y_1, \cdots, y_{t-1}}.
\end{equation}
which leads to:
$$
\Var{y_t\,|\, y_1, \cdots, y_{t-1}} =
\left(\frac{ \lambda_t \beta_{t-1} }
{\alpha_{t-1}-1}  \right)^2
\left[
\frac{\psi+1}{q\left(\alpha_{t-1}-2\right)}
+\psi
\right].
$$

Finally, the following lemma is an analog of Lemma 3.

\begin{lemma}\label{cor.a2}

Under Model \ref{mod-1}, we have
the one-step-ahead forecasting of the severity at time $\tau$ given by
    \begin{equation}\label{eq.3}
    \E{y_\tau\,|\, y_1, \cdots, y_{t-1}}=\lambda_{\tau}
    \left[
    b_0^* + \sum\limits_{t=1}^{\tau-1}b_t^* \frac{y_t}{\lambda_t}
    \right]
    \end{equation}
where
\begin{equation*}
    b_0^*:=\frac{\beta_0}{\alpha_{\tau-1}-1}b_0^{**}
\quad\hbox{and}\quad
    b_t^*:=
    \frac{1}{\left(\alpha_{\tau-1}-1\right) \psi}b_t^{**}
    \end{equation*}
    for $t=1, \cdots, \tau-1$, where
\begin{equation*}
b_t^{**}:=
\begin{cases}
  \prod\limits_{k=t+1}^{\tau}q_k^*, & t=0,1, \cdots, \tau-1;\\
  1, &t=\tau.
\end{cases}
\end{equation*}

\end{lemma}

\begin{remark}
In Lemma \ref{lem.a2}, it is interesting to observe that the updating rule in \eqref{eq.2} of Model \ref{mod-1} with time-varying $q_t$ and $q_t^*$ guarantees $\alpha_t>1$ so that the following expectation
\[
\E{\theta_{t} \,|\, y_1, \cdots, y_{t-1}}<\infty
\]
is well defined. However, because $q_t$ and $q_t^*$ are time-varying, the weight $b_t^*$ in \eqref{eq.3} of  Model \ref{mod-1} is not exponentially decaying unlike the weight $q^{\tau-1-t}$ in \eqref{eq.4} of Model \ref{mod.0}.


\end{remark}

While it is natural to put more emphasize the recent claims than the old claims in the posterior ratemaking process, the model with the static random effect model as in \citet{lee2020poisson} cannot distinguish the seniority of the claims and put the same weights to all claims regardless of their seniority. On the other hand, Model \ref{mod.0} and Model \ref{mod-1} distinguish the seniority of the claims with the aid of dynamically updated state variables. In the following, we show that the one step-ahead forecasting can be represented as the linear combination of claims with timely ordered weights. Note that, for the pair comparison of the contribution of each claims, it is important to assume the same prior rates \citep{ahn2021order}.

\begin{corollary}\label{cor.a3}
    Under Model \ref{mod.0},
    if we assume
    \[
    \lambda_1=\cdots=\lambda_{\tau-1},
    \]
    then we have the following order:
    \[
    b_1<\cdots<b_{\tau-1}.
    \]
    Similarly, under Model \ref{mod-1},
    if we assume
    \[
    \lambda_1=\cdots=\lambda_{\tau-1},
    \]
    then we have the following order:
    \[
    b_1^*<\cdots<b_{\tau-1}^*.
    \]

\end{corollary}

\subsection{Alternative specification of $q_t$, $q_t^*$}

The transition equation of $\theta_t$ in Model \ref{mod-1} allows the same conditional expectation and variance as those of Model \ref{mod.0} which leads to the intuitive updating rule of the random effect; see \eqref{up.1}, \eqref{up.2}, \eqref{eq.c1}, and \eqref{eq.c2}. However, unlike Model \ref{mod.0}, Lemma \ref{cor.a2} shows that this transition equation does not lead to the exponentially weighted moving average (EWMA). Let us now investigate whether alternative specifications for $q_t$ and $q_t^*$ would be possible, so that  EWMA formula can be obtained.

Let us now replace, in Model \ref{mod-1}, specification \eqref{qt2} by the following assumption:
\begin{equation}
	\label{alternativeqt}
	q_t:=\frac{q(\alpha_{t-1}-1)+1}{\alpha_{t-1}}\quad\hbox{and}\quad q_t^*:=q, \quad \hbox{for}\quad q\in(0,1).
\end{equation}
For $\alpha_0>1$, it is easy to show
\[
q_t\in(0,1) \quad\hbox{and}\quad q_t^*>0 \quad\hbox{for}\quad t=1,2,\cdots
\]
so that the specifications in \eqref{alternativeqt} are valid in Model \ref{mod-1}.
Furthermore, simple algebraic calculation shows that
\[
q_t\alpha_{t-1}>1, \quad t=1, 2, \cdots
\]
as long as $\alpha_0>1$.
As a result, Lemma \ref{lem.a2} deduces that the martingale condition \eqref{eq.c1} still holds. In other words, the predictive mean always exists and simple algebra leads to
\[
\E{y_\tau\,|\, y_1, \cdots, y_{t-1}}=\lambda_{\tau}\frac{q^{\tau-1}\beta_0 + q^{\tau-2}\frac{1}{\psi}\frac{y_1}{\lambda_1}+\cdots+\frac{1}{\psi}\frac{y_{\tau-1}}{\lambda_{\tau-1}}}
{q^{\tau-1}(\alpha_0-1)+\frac{1-q^{\tau-1}}{1-q}\frac{1}{\psi}}
\]
which is EWMA in $y_1,\cdots,y_{\tau-1}$.

However, from Lemma \ref{lem.5}, it is clear that
the variance formula in \eqref{eq.c2} no longer holds. We first investigate the existence of the conditional variance of $\theta_t \vert y_{t-1},\cdots,y_1$, which is equivalent to check whether we have $q_t \alpha_{t-1}>2$ for any time $t$.
Combining \eqref{parameterupdating} and \eqref{alternativeqt}, we get
%
\begin{equation}
	\label{qalpha}
q_{t}\alpha_{t-1}=q q_{t-1}\alpha_{t-2} +q \left(\frac{1}{\psi}-1\right)+1.
\end{equation}
Thus if $q_{t-1}\alpha_{t-2}>2$, then $q_{t}\alpha_{t-1}>q (\frac{1}{\psi}+1)+1$.

As a consequence, if we have:
\begin{equation}\label{eq.a1}
q_1\alpha_0 >2
\end{equation}
and
\[
	q \left(\frac{1}{\psi}+1\right)+1\geq 2
\]
then the predictive variance of $\theta_t$ always exists at any time $t$.\footnote{Note that the condition in \eqref{eq.a1} is equivalent with
\[
q(\alpha_0-1)>1.
\]}
Furthermore, in such case, we have
\[
\Var{\theta_\tau | y_1, \cdots, y_{\tau-1}}=
\frac{\alpha_{\tau-1}-2}{q(\alpha_{\tau-1}-1)-1}
\Var{\theta_{\tau-1} | y_1, \cdots, y_{\tau-1}}
\]
with
\[
\frac{\alpha_{\tau-1}-2}{q(\alpha_{\tau-1}-1)-1}>1.
\]
This alternative specification will be used later on to construct an alternative specification of our main bivariate frequency-severity model.

\section{The dynamic random effect based frequency-severity model}

In this section, we combine Model \ref{mod.0} and Model \ref{mod-1} to construct a  dynamic frequency-severity model. This model allows for dependence between the frequency and severity variables, by inspiring from \citet{garrido2016generalized}.

\subsection{The proposed model}
\begin{model}\label{mod.1}
For constants $q^{[1]}, q^{[2]}\in(0,1)$, consider the stochastic process
\[
\left(\boldsymbol{y}_t, \boldsymbol{\theta}_t \right)_{t\in\mathbb{N}}
\]
where the bivariate observations $\boldsymbol{y}_t:=(y_t^{[1]}, y_t^{[2]})$ are frequency-severity type and state variables (or random effects) $\boldsymbol{\theta}_t:=(\theta_t^{[1]}, \theta_t^{[2]})$ are bivariate positive, with the following joint dynamics
\begin{enumerate}
  \item[i. ] (Initial condition for the state variable) At the initial date, the two random effects are mutually independent, and gamma and inverse gamma distributed, respectively:
  \[
  \left(\theta^{[1]}_0, \theta^{[2]}_0 \right)\sim \prod\left(F_{\rm Gamma}\left(\cdot; \alpha^{[1]}_0, \beta^{[1]}_0\right), F_{\rm Gamma}\left(\cdot; \alpha^{[2]}_0, \beta^{[2]}_0\right)\right)
  \]
	where parameters $\alpha_0^{[1]}, \beta_0^{[1]}, \alpha_0^{[2]}, \beta_0^{[2]}$ are positive, with $\alpha_0^{[2]}>1$.

 \item[ii.] (Transition equation at time $t$): Assume that at a certain time $t$,  the filtering distribution $ \left(\theta_{t-1}^{[1]}, \theta_{t-1}^{[2]}\right)|\mathcal{F}_{t-1} $ is independent gamma and inverse gamma distribution
 \begin{equation}\label{eq.34}
 \left(\theta_{t-1}^{[1]}, \theta_{t-1}^{[2]}\right)|\mathcal{F}_{t-1} \sim \prod\left(F_{\rm Gamma}\left( \cdot ;  \alpha_{t-1}^{[1]}, \beta_{t-1}^{[1]}\right), F_{\rm IG}\left(\cdot;\alpha_{t-1}^{[2]}, \beta_{t-1}^{[2]}\right)\right)
 \end{equation}
with parameters $\alpha_{t-1}^{[1]}$, $\beta_{t-1}^{[2]}$, $\alpha_{t-1}^{[2]}$, and $\beta_{t-1}^{[2]}$ being the functions of $\mathcal{F}_{t-1}$, and $\alpha_{t-1}^{[2]}>1$, and let
$\left(B_t^{[1]}, B_t^{[2]}\right)$ be conditionally independent with $\left(\theta_{t-1}^{[1]}, \theta_{t-1}^{[2]} \right)$ given $\mathcal{F}_{t-1}$, with independent beta marginals:
 \[
 \left(B_t^{[1]}, B_t^{[2]}\right)|\mathcal{F}_{t-1}\sim \prod\left(F_{\rm Beta}\left( \cdot ; q^{[1]} \alpha_{t-1}^{[1]}, (1-q^{[1]}) \alpha_{t-1}^{[1]}\right), F_{\rm Beta}\left(\cdot; q^{[2]}_t \alpha_{t-1}^{[2]}, (1-q^{[2]}_t) \alpha_{t-1}^{[2]}\right)\right)
 \]
 for
 \begin{equation}\label{eq.42}
 q^{[2]}_t := \frac{q^{[2]}(\alpha_{t-1}^{[2]}-2)+2}{\alpha_{t-1}^{[2]}}.
 \end{equation}
 Then, the new state variables $(\theta_t^{[1]}, \theta_t^{[2]})$ are defined by
 \begin{equation}\label{eq.36}
(\theta_t^{[1]}, \theta_t^{[2]}):= \left(\theta_{t-1}^{[1]}\frac{B_t^{[1]}}{q^{[1]}}, \theta_{t-1}^{[2]}\frac{q^{*[2]}_t}{B_t^{[2]}} \right).
\end{equation}
where\footnote{The motivation for the definitions of $q_{t}^{[2]}$ and $q_{t}^{*[2]}$ are given in Corollary \ref{cor.1}.}
\begin{equation}\label{eq.43}
q^{*[2]}_t:=\frac{q^{[2]}\left( \alpha_{t-1}^{[2]}-2\right)+1}{\alpha_{t-1}^{[2]}-1}.
\end{equation}
\item[iii.] (Observed variable at time $t$) For some given positive constants
\[
\lambda_t^{[1]}, \quad \lambda_t^{*[2]}, \psi^{[2]}, \quad\hbox{and}\quad \eta
\]
which might depend on individual characteristics, the observation at time $t$ $\left( y_t^{[1]}, y_t^{[2]}\right)$ is drawn from the following conditional distribution:
    \begin{equation}\label{eq.38}
    y_t^{[1]} | \mathcal{F}_{t-1}, \mathcal{F}_{t}^{\rm [state]} \sim {\rm Pois}\left( \lambda_t^{[1]}\theta_t^{[1]}\right)
    \end{equation}
    and
    \begin{equation}\label{eq.39}
    y_t^{[2]} | \mathcal{F}_{t-1}, \mathcal{F}_{t}^{\rm [state]}, y_t^{[1]} \sim {\rm Gamma}\left( \frac{y_t^{[1]}}{\psi^{[2]}}, \frac{1}{\theta_t^{[2]}\lambda_t^{[2]}\psi^{[2]}}\right)
    \end{equation}
where
\[\lambda_t^{[2]}:=\lambda_t^{*[2]}\exp\left( \eta\, y_t^{[1]}\right).\]
\end{enumerate}
\end{model}

 Thus, starting from the initial condition, the transition equations and observation equations are applied in sequential order as in Figure \ref{chainstructure0}.
\begin{figure}[h!]
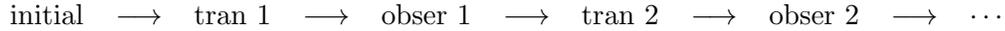

\centering
\begin{tabular}{ccccccccccc}
initial & $\longrightarrow$  & tran $1$ & $\longrightarrow$ & obser $1$ & $\longrightarrow$ & tran $2$ & $\longrightarrow$ & obser $2$ & $\longrightarrow$ & $\cdots$
\end{tabular}
\caption{Chain of transitions and observations equations: initial, tran $t$, obser $t$ refer to the initial condition (part $i$), transition process at time $t$ (part $ii$), and observation process at time $t$ (part $iii$), respectively.}\label{chainstructure0}
\end{figure}

We can also remark that, in the above model, if we focus only on the joint dynamics of $ y_t^{[1]}$ and $\theta_t^{[1]}$, then we get the standard HF model for time series of counts. In particular, the transition equation from $\theta_{t-1}^{[1]}$ to $\theta_t^{[1]}$ does not depend on the severity component $y_{t-1}^{[1]}$. In other words, their joint dynamics is \textit{exogenous} vis-\`a-vis the severity process and its random effect process. The dynamics of this latter, however, depends on $\theta_{t-1}^{[1]}$ through the conditional gamma distribution \eqref{eq.39}, whose shape and scape parameters depend both on $y_t^{[1]}$. The fact that its shape parameter is a multiple of $y_t^{[1]}$ means roughly that the $y_t^{[1]}$ individual claim amounts are i.i.d. gamma distributed, with the dispersion parameter $\phi^{[2]}$. This assumption is standard in the insurance literature (\cite{gourieroux1999econometrics}). The specification of the shape parameter in equation  \eqref{eq.39} can be viewed as a (dynamic random effect based) extension of the model of \cite{garrido2016generalized}, in the sense that the claim count enters into the regression model of (individual) claim amounts. 
Furthermore, the shape parameter depends also on $\theta^{[2]}_t$, which is defined through $\theta^{[2]}_{t-1}$, and the latter depends in turn on $y_{t-1}^{[1]}$ and so on. As a consequence,  the conditional distribution of $y_t^{[2]}$ given $y_t^{[1]}$ and $\mathcal{F}_{t-1}$ depends on all the past frequency and severity variables.  



The following theorem, which is an analog of Lemma 2  for HF model and Lemma 5 for SM model, ensures that the gamma distributional assumptions in \eqref{eq.34} are indeed satisfied:

 \begin{theorem} [Recursion for forecasting and filtering]\label{thm.1}
We have
 \begin{enumerate}[$i.$]
 \item One-step-ahead forecasting of the random effects at each time $t$ is given by
 \begin{equation}\label{eq.41}
 \left(\theta_{t}^{[1]}, \theta_{t}^{[2]}\right)|\mathcal{F}_{t-1} \sim \prod\left(F_{\rm Gamma}\left( \cdot;  q^{[1]}\alpha_{t-1}^{[1]}, q^{[1]}\beta_{t-1}^{[1]}\right), F_{\rm IG}\left( \cdot; q_t^{[2]}\alpha_{t-1}^{[2]}, q_t^{*[2]}\beta_{t-1}^{[2]}\right)\right).
 \end{equation}
 where $q_t^{[2]}$ and $q_t^{*[2]}$ is given in \eqref{eq.42} and \eqref{eq.43}.
    \item The one-step-ahead forecasting density for the observations is given by
\begin{equation}\label{eq.51}
y_t^{[1]}\,|\, \mathcal{F}_{t-1} \sim {\rm NB}\left( \lambda_t^{[1]}\frac{\alpha_{t-1}^{[1]}}{\beta_{t-1}^{[1]}}, q^{[1]}\alpha_{t-1}^{[1]}\right)
\end{equation}
and
\begin{equation}\label{eq.52}
y_t^{[2]}\,|\, \mathcal{F}_{t-1}, y_t^{[1]} \sim {\rm GB2}\left(
1\,,\,
{q_t^{*[2]}}\beta_{t-1}^{[2]}\lambda_{t}^{[2]}\psi^{[2]} \,,\,
\frac{y_t^{[1]}}{\psi^{[2]}}\,,\, 
{q_t^{[2]}  \alpha_{t-1}^{[2]}}
\right)
\end{equation}
with
\[
\E{y_t^{[2]}\,|\, \mathcal{F}_{t-1}, y_t^{[1]}}=y_t^{[1]}\lambda_t^{[2]}\frac{\beta_{t-1}^{[2]}}{\alpha_{t-1}^{[2]}-1}.
\]
\item The filtering distribution at time $t$ is given as
    \begin{equation}\label{eq.35}
 \left(\theta_{t}^{[1]}, \theta_{t}^{[2]}\right)|\mathcal{F}_t \sim \prod\left(F_{\rm Gamma}\left( \cdot; \alpha_{t}^{[1]}, \beta_{t}^{[1]}\right), F_{\rm IG}\left( \cdot; \alpha_{t}^{[2]}, \beta_{t}^{[2]}\right)\right)
 \end{equation}
 with
   \begin{equation}\label{eq.45}
   \begin{cases}
     \alpha_{t}^{[1]}:=q^{[1]}\alpha_{t-1}^{[1]}+y_t^{[1]}\\
     \beta_{t}^{[1]}:=q^{[1]}\beta_{t-1}^{[1]}+\lambda_t^{[1]}\\
     \alpha_{t}^{[2]}:= {q_t^{[2]}\alpha_{t-1}}
     +\frac{y_t^{[1]}}{\psi^{[2]}}\\
     \beta_{t}^{[2]}:={q_t^{*[2]}}
     \beta_{t-1}^{[2]}+\frac{y_t^{[2]}}{\lambda_t^{[2]}\psi^{[2]}}
   \end{cases}
   \end{equation}
   with $\alpha_{t}^{[1]}, \beta_{t}^{[1]}, \beta_{t}^{[2]}>0$ and $\alpha_{t}^{[2]}>1$.

\item
If $\alpha_{t-1}^{[2]}>2$, then since
$
q^{[2]}\left( \alpha_{t-1}^{[2]}-2\right)+2>2,
$
we have
$$
\Var{y_t^{[2]}\,|\, \mathcal{F}_{t-1}, y_t^{[1]}} =
y_t^{[1]}\left(\frac{ \lambda_t^{[2]} \beta_{t-1}^{[2]} }
{\alpha_{t-1}^{[2]}-1}  \right)^2
\left[
\frac{\psi^{[2]}+y_t^{[1]}}{q^{[2]}\left(\alpha_{t-1}^{[2]}-2\right)}
+\psi^{[2]}
\right].
$$

   \item If we further assume that $\alpha_{t-1}^{[2]}>2$, then we have $\alpha_{t}^{[2]}>2$. In other words, the shape parameter of the gamma distribution remains larger than 2 when $t$ increases, so long as its initial value is larger than 2.
 \end{enumerate}
 \end{theorem}
 \begin{proof}


Part $i$ is a consequence of  \eqref{eq.34}, \eqref{eq.36} and Lemma \ref{lem.1} and Lemma \ref{lem.5}.
Part $ii$ is a direct consequence of part $i$, and the fact that the negative binomial distribution is the Poisson mixture with the gamma distribution, and  GB2 distributions are the gamma mixture with inverse gamma distribution. Part $iii$ is a consequence of part $ii$ and the variance formula for the GB2 distribution (see Section \ref{sec.2.1}).
Part $iv$ is a consequence of the Poisson-Gamma and Gamma-inverse Gamma conjugacy. More precisely,
 \[
 \begin{aligned}
 \pi\left(\theta_{t}^{[1]}, \theta_{t}^{[2]} \,|\, y_t^{[1]}, y_t^{[2]}, \mathcal{F}_{t-1}\right)
 &\propto
 f\left(y_t^{[1]}, y_t^{[2]} \,|\, \theta_{t}^{[1]}, \theta_{t}^{[2]}, \mathcal{F}_{t-1}\right)
 \pi\left(\theta_{t}^{[1]}, \theta_{t}^{[2]} \,|\, \mathcal{F}_{t-1} \right)\\
 &\propto
 \left( \theta_{t}^{[1]} \right)^{q^{[1]}\alpha_{t-1}^{[1]}+y_t^{[1]}}
 \exp\left( -\theta_{t}^{[1]}\left(  q^{[1]} \beta_{t-1}^{[1]}+\lambda_t^{[1]}\right)\right)\\ 
 &\quad\quad
 \left( 1/\theta_{t}^{[2]} \right)^{q_t^{[2]}\alpha_{t-1}^{[1]} + y_{t}^{[1]}/\psi^{[2]} -1}
  \exp\left( -\frac{1}{\theta_{t}^{[2]}}\left( q_t^{*[2]}\beta_{t-1}^{[2]}+\frac{y_t^{[2]}}{\lambda_t^{[2]}\psi^{[2]}}\right)\right)\\ 
 \end{aligned}
 \]
 which implies
 \[
 \begin{aligned}
 &\left(\theta_{t}^{[1]}, \theta_{t}^{[2]}\right)\,|\, y_t^{[1]}, y_t^{[2]}, \mathcal{F}_{t-1} \sim
 \prod\left({\rm Gamma}\left( \alpha_{t}^{[1]}, \beta_{t}^{[1]}\right), {\rm IG}\left(\alpha_{t}^{[2]}, \beta_{t}^{[2]}\right)\right)
 \end{aligned}
 \]
 where $\alpha_{t}^{[1]}, \beta_{t}^{[1]}, \alpha_{t}^{[2]}, \beta_{t}^{[2]}$ are given in \eqref{eq.45}. Hence, we have finished the proof of part iv.
 Finally, part $v$ is a direct consequence of part $iii$ and the definition of $\alpha_{t}^{[2]}$ given in equation \eqref{eq.45}.
 \end{proof}

We note that the closed form expression of the likelihood function can be obtained from  the one-step-ahead forecasting density for the observations in Theorem \ref{thm.1}.
The following result 
intuitively explains the dynamics of the random effects which preserve the mean while inflating the variance in one-step-ahead forecasting of the random effects. 

\begin{corollary}\label{cor.1}
    Under Model \ref{mod.1}, for a positive constants $\alpha_0^{[1]}, \beta_0^{[1]}, \alpha_0^{[2]}, \beta_0^{[2]}$ with $\alpha_0^{[2]}>1$, we have the following results.
    \begin{enumerate}
      \item[i.]  We have the following relation of conditional mean
      \[
      \E{\theta_{t}^{[1]} \,|\, \mathcal{F}_{t-1}}=
      \E{\theta_{t-1}^{[1]} \,|\, \mathcal{F}_{t-1}}
      \]
      and
      \[
      \E{\theta_{t}^{[2]}\,\bigg|\, \mathcal{F}_{t-1}}=
      \E{\theta_{t-1}^{[2]}\,\bigg|\, \mathcal{F}_{t-1}}.
      \]

      \item[ii.] We have the following relation of the conditional variance
      \[
      \Var{\theta_{t}^{[1]} \,|\, \mathcal{F}_{t-1}}=\frac{1}{q^{[1]}}
      \Var{\theta_{t-1}^{[1]} \,|\, \mathcal{F}_{t-1}}.
      \]
      Furthermore, if $\alpha_0^{[2]}>2$, we have
      \[
      \Var{\theta_{t}^{[2]}\,\bigg|\, \mathcal{F}_{t-1}}=\frac{1}{q^{[2]}}
      \Var{\theta_{t-1}^{[2]}\,\bigg|\, \mathcal{F}_{t-1}}.
      \]

    \end{enumerate}
\end{corollary}
Because
$q^{[1]}, q^{[2]}\in(0,1)$, the factors $1/q^{[1]}$ and $1/q^{[2]}$ can be interpreted as variance inflating factors of frequency and severity, respectively.

\subsection{Posterior ratemaking}
This section provides the analytical expression of the posterior mean of the aggregate severity under Model \ref{mod.1}.

\begin{theorem}\label{thm.2}
 In Model \ref{mod.1}, for positive constants $\alpha_0^{[1]}, \beta_0^{[1]}, \alpha_0^{[2]}, \beta_0^{[2]}$ with $\alpha_0^{[2]}>1$.
     \item[i.] The one-step-ahead forecasting of the frequency at time $\tau$ is expressed as
      \[
      \begin{aligned}
    \E{y_\tau^{[1]} \,|\, \mathcal{F}_{\tau-1} }
    &=\lambda_\tau^{[1]}
\left[
\omega_0^{[1]}+\sum\limits_{t=1}^{\tau-1}\omega_t^{[1]}\frac{y_t^{[1]}}{\lambda_t^{[1]}}
\right]\\
\end{aligned}
\]
where
\[
\omega_0^{[1]}:=\frac{\left(q^{[1]}\right)^{\tau-1}\alpha_0^{[1]}}{\sum\limits_{t=1}^{\tau-1}\left(q^{[1]}\right)^{\tau-1-t}\lambda_{t}^{[1]} +
    \left(q^{[1]}\right)^{\tau-1} \beta_0}
\quad\hbox{and}\quad
\omega_t^{[1]}:=\frac{\left(q^{[1]}\right)^{\tau-1-t}\lambda_{t}^{[1]}}{\sum\limits_{t=1}^{\tau-1}
\left(q^{[1]}\right)^{\tau-1-t}\lambda_{t}^{[1]} +
    \left(q^{[1]}\right)^{\tau-1} \beta_0}
\]
for $t=1, \cdots, \tau-1$.
      If we further assume
    \begin{equation}
    \label{eq.cond.1}
    \eta < \log\left( \frac{q^{[1]}\beta_{\tau-1}^{[1]}+\lambda_\tau^{[1]}}{\lambda_\tau^{[1]}}\right),
    \end{equation}
      then we have the following results.
    \begin{enumerate}
    \item[ii.] The one-step-ahead forecasting of the aggregate severity at time $t$ is expressed as
    \[
    \E{y_\tau^{[2]}\,|\, \mathcal{F}_{\tau-1}}=\lambda_{\tau}^{*[2]}\E{\exp\left( \eta y_\tau^{[1]}\right)y_\tau^{[1]} \, | \, \mathcal{F}_{\tau-1}}
    \left[
    \omega_0^{[2]} + \sum\limits_{t=1}^{\tau-1}\omega_t^{[2]} \frac{y_t^{[2]}}{\lambda_t^{[2]}}
    \right]
    \]
    where the expectation on the right hand side can be computed as:
 \begin{equation}
 \label{eq.cond.2}
 \begin{aligned}
     &\E{\exp\left( \eta y_\tau^{[1]}\right)y_\tau^{[1]} \, | \, \mathcal{F}_{\tau-1}}\\
     &=q^{[1]}\alpha_{\tau-1}^{[1]}\lambda_{\tau}^{[1]}e^{\eta}
    \left(
    \frac{q^{[1]}\beta_{\tau-1}^{[1]}}
    {\lambda_{\tau}^{[1]} + q^{[1]}\beta_{\tau-1}^{[1]}-\lambda_{\tau}^{[1]}e^{\eta}}
    \right)^{q^{[1]}\alpha_{\tau-1}^{[1]}-1}
    \frac{1}
    {\left(\lambda_{\tau}^{[1]} + q^{[1]}\beta_{\tau-1}^{[1]}-\lambda_{\tau}^{[1]}e^{\eta}\right)} .
    \end{aligned}
    \end{equation}
    and
    \begin{equation}\label{eq.62}
    \omega_0^{[2]}:=\frac{\beta_0^{[2]}}{\alpha_{\tau-1}^{[2]}-1}\omega_0^{*[2]}
\quad\hbox{and}\quad
    \omega_t^{[2]}:=
    \frac{1}{\left(\alpha_{\tau-1}^{[2]}-1\right)\psi^{[2]}}\omega_t^{*[2]}
    \end{equation}
    for $t=1, \cdots, \tau-1$, where
\begin{equation}\label{eq.61}
\omega_t^{*[2]}:=
\begin{cases}
  \prod\limits_{k=t+1}^{\tau}q_k^{*[2]}, & t=0,1, \cdots, \tau-1;\\
  1, &t=\tau.
\end{cases}
\end{equation}

\end{enumerate}
\end{theorem}

\begin{remark}
    From the recursion \eqref{eq.45}, we can be easily check that
 for any $\tau=1, 2, \cdots$.
\begin{equation}
\label{transformrecursion}
\begin{cases}
  \alpha_{\tau}^{[1]}=\left( q^{[1]}\right)^{\tau}\alpha_{0}^{[1]}+ \sum\limits_{t=1}^{\tau}
  \left( q^{[1]}\right)^{\tau-t}y_t^{[1]};\\
  \beta_{\tau}^{[1]}=\left( q^{[1]}\right)^{\tau}\beta_{0}^{[1]}+ \sum\limits_{t=1}^{\tau}
  \left( q^{[1]}\right)^{\tau-t}\lambda_t^{[1]};\\
  \alpha_{\tau}^{[2]}=\left( q^{[2]}\right)^{\tau}\left( \alpha_{0}^{[2]}-2\right)+
  \sum\limits_{t=1}^{\tau} \left( q^{[2]}\right)^{\tau-t}\frac{y_t^{[1]}}{\psi^{[2]}}+2;\\
  \beta_{\tau}^{[2]}=\omega_0^{*[2]} \beta_0^{[2]}+\sum\limits_{t=1}^{\tau} \omega_t^{*[2]} \frac{y_t^{[2]}}{\lambda_t^{[2]}\psi^{[2]}}
\end{cases}
\end{equation}
where $\omega_t^{*[2]}$'s are defined in \eqref{eq.61}.

\end{remark}

\begin{proof}
Part $i$, which is first derived in HF,  is obtained by simply using the first two recursions of \eqref{eq.45} to express $\alpha_t^{[1]}$ and $\beta_t^{[1]}$ in terms of $(\lambda_s)_s$ and $(y_s^{[1]})_s$.

  For part $ii$, from Theorem \ref{thm.1}, we have
    \begin{equation}\label{eq.620}
  \begin{aligned}
    \E{y_\tau^{[2]}\,|\, \mathcal{F}_{\tau-1}}&=\E{\E{y_\tau^{[2]} \,|\, \mathcal{F}_{\tau-1}, y_\tau^{[1]}}\,|\, \mathcal{F}_{\tau-1}}\\
    &=\frac{\lambda_{\tau}^{*[2]}}{\alpha_{\tau-1}^{[2]}-1}
    \E{\exp\left( \eta y_\tau^{[1]}\right)y_\tau^{[1]} \, | \, \mathcal{F}_{\tau-1}}
    \beta_{\tau-1}^{[2]}\\
    &=\frac{\lambda_{\tau}^{*[2]}}{\alpha_{\tau-1}^{[2]}-1}
    \E{\exp\left( \eta y_\tau^{[1]}\right)y_\tau^{[1]} \, | \, \mathcal{F}_{\tau-1}}
    \left(
    \omega_0^{*[2]} \beta_0^{[2]}+\sum\limits_{t=1}^{\tau-1} \omega_t^{*[2]} \frac{y_t^{[2]}}{\lambda_t^{[2]}\psi^{[2]}}
    \right).
  \end{aligned}
  \end{equation}
  where the last equality is from the above equation \eqref{transformrecursion}.  Finally, equation \eqref{eq.cond.2} is a consequence of the Laplace transform formula of the negative binomial distribution.

\end{proof}

\begin{remark}
Assumption in \eqref{eq.cond.1} is necessary to ensure that the expectation in \eqref{eq.cond.2} is finite, so that the total claim amount is of finite mean. In other words, the dependence between frequency and severity components should either be negative ($\eta<0$), or not excessively positive. This assumption is typically satisfied in auto insurance, since empirical studies often find negative, or weakly positive frequency-severity dependence \citep{shi2015dependent, park2018does, lu2019flexible}.


\end{remark}

The following theorem is an analog of Corollary 1 and shows that the one step-ahead forecasting can be represented as the linear combination of claims of the aggregate severities with timely ordered weights.

\begin{theorem}
      Consider the settings in Model \ref{mod.1}, for a positive constant $\alpha_0^{[1]}, \beta_0^{[1]}, \alpha_0^{[2]}, \beta_0^{[2]}$ with $\alpha_0^{[2]}>1$. If we further assume assumption in \eqref{eq.cond.1}, and
      \[
      \lambda_1^{[1]}=\cdots=\lambda_{\tau-1}^{[1]}\quad\hbox{and}\quad \lambda_1^{[2]}=\cdots=\lambda_{\tau-1}^{[2]},
      \]
       then we have the following results:
    \[
    \omega_1^{[1]}<\cdots< \omega_{\tau-1}^{[1]}\quad\hbox{and}\quad
    \omega_1^{[2]}<\cdots< \omega_{\tau-1}^{[2]}.
    \]
\end{theorem}
\begin{proof}
First, the following inequality
\[
\omega_1^{[1]}<\cdots< \omega_{\tau-1}^{[1]}
\]
is immediate result of Corollary \ref{cor.a3}.
Now, we prove
\[
\omega_1^{[2]}<\cdots< \omega_{\tau-1}^{[2]}.
\]
    From Theorem \ref{thm.1}, we know that $\alpha_0^{[2]}>1$ implies $\alpha_t^{[2]}>1$ for $t=1, 2, \cdots$, which further implies
    \begin{equation}\label{eq.63}
    q_t^{*[2]}<1, \quad t=1, 2, \cdots.
    \end{equation}
which, together with the definition of $\omega_t^{*[2]}$ in \eqref{eq.61}, implies
    \[
    \omega_1^{*[2]}<\cdots<\omega_\tau^{*[2]}.
    \]
    Hence
        \[
    \omega_1^{[2]}<\cdots< \omega_{\tau-1}^{[2]}.
    \]
\end{proof}

\section{Alternative updating rules of the frequency-severity model}
The major appeal of SM and HF models is that they have similar, simple updating rule for the random effect [see \eqref{up.1}, \eqref{up.2}, \eqref{eq.c1}, and \eqref{eq.c2}]. However, to emphasize on the flexibility of the proposed Bayesian state-space model, this section discusses various alternative transition equations of the random effect which leads to different updating rules of the random effect, and hence forecasting formula.

\subsection{Exponentially weighted moving average for severity part}
Just as the SM model specified in section 2.3.2, Model \ref{mod.1} does not allow EWMA predictive formula for the severity variable. Let us therefore investigate how the alternative specification proposed in section 2.3.3. can be adapted to the frequency-cost model.




More precisely, we replace the specification of $q_t^{[2]}$ in \eqref{eq.42} and that of $q_t^{*[2]}$ in \eqref{eq.43} by
\[
q_t^{[2]}:=\frac{q^{[2]}(\alpha_{t-1}^{[2]}-1)+1}{\alpha_{t-1}^{[2]}}\quad\hbox{and}\quad
q^{*[2]}_t=q^{[2]},
\]
respectively. Then the forecasting formula of the frequency variable remains unchanged. As for the aggregate severity, the one-step-ahead forecasting of the aggregate severity at time $\tau$ is given by
\[
\begin{aligned}
    &\E{y_\tau^{[2]}\,|\, \mathcal{F}_{\tau-1}}\\
    &=\lambda_{\tau}^{*[2]}\E{\exp\left( \eta y_\tau^{[1]}\right)y_\tau^{[1]} \, | \, \mathcal{F}_{\tau-1}}
    \left[
    \frac{\left( q^{[2]}\right)^{\tau-1}\beta_0^{[2]}+\sum\limits_{t=1}^{\tau-1}
    \left( q^{[2]}\right)^{\tau-1-t}\frac{y_t^{[2]}}{\lambda_t^{[2]}\psi^{[2]}}
    }
    {\left( q^{[2]}\right)^{\tau-1}\alpha_0^{[2]}-\left( q^{[2]}\right)^{\tau-1}+
    \frac{1}{\psi^{[2]}}\sum\limits_{t=1}^{\tau-1}
    \left( q^{[2]}\right)^{\tau-1-t}y_t^{[1]}
    }
    \right]
\end{aligned}
\]
in case \eqref{eq.cond.1} is satisfied.
Clearly, the one-step-ahead forecasting of the aggregate severity has EWMA.

One extra difficulty of adapting the variant of SM model to the frequency-severity model is the existence of the one-step-ahead variance of the random effect unlike Corollary \ref{cor.1} in case of Model \ref{mod.1}.
Specifically, an analog of equation \eqref{qalpha} is:
$$
q^{[2]}_{t}\alpha_{t-1}^{[2]}=q^{[2]} q^{[2]}_{t-1}\alpha_{t-2}^{[2]} +q^{[2]} (\frac{y^{[1]}_{t}}{\psi^{[2]}}-1)+1.
$$
Thus, when $y^{[1]}_{t}$'s take values of zeros, even if
$$q^{[2]}_{t-1}\alpha_{t-2}^{[2]} >2, $$
the sequence $(q^{[2]}_{t}\alpha_{t-1})$ can take value smaller than $2$ for properly large $t$, which leads to infinite variance of $\theta_t^{[2]}$ conditional on $\mathcal{F}_t$ unlike in Corollary \ref{cor.1}.
 Intuitively, this is due to the fact in this variant of Model 3, $(\theta^{[2]}_t)$ is assumed to vary, even during a period when no claim is observed. It is shown in the next subsection that this ``forced" evolution also exists in the plain Model 3, which motivates a second variant proposed below.


\subsection{A three-part variant}
The frequency-severity model has the characteristics that for those periods where no claims are reported $(y_t^{[1]}=0)$, the claim amount $y_t^{[2]}$ is automatically zero.  In other words, the individual claim size can be viewed as \textit{unobserved}. One of the biggest advantages of state-space models is that they are very convenient, and flexible, to deal with missing observations. The aim of this section is to explain how this flexibility can be easily explored to adapt Model 3 and allow for different treatments of these non-claim periods.

 First, remark that in Model 3 even if $y_t^{[1]}=0$, the random effect $\theta_t^{[2]}$ will still be updated to $\theta_{t+1}^{[2]}$, with a different predictive distribution.  
Indeed, by equation \eqref{eq.41}, these predictive distribution are:
$$
\theta_t^{[2]} \vert \mathcal{F}_{t-1} \sim F_{\rm Gamma}\left( \cdot; q_t^{[2]}\alpha_{t-1}^{[2]}, q_t^{*[2]}\beta_{t-1}^{[2]}\right),
$$
and
$$
\theta_{t+1}^{[2]} \vert \mathcal{F}_{t} \sim F_{\rm Gamma}\left( \cdot; q_{t+1}^{[2]}\alpha_{t}^{[2]}, q_{t+1}^{*[2]}\beta_{t}^{[2]}\right),
$$
respectively. Then we have:
\begin{align*}
q_{t+1}^{[2]}\alpha_{t}&=q^{[2]}(\alpha_{t}^{[2]}-2)+2 \qquad  \qquad \text{     (by equation \eqref{eq.42}) } \\
&= q^{[2]}(q^{[1]}\alpha_{t-1}^{[2]}-2)+2  \qquad  \qquad \text{     (by equation \eqref{eq.45}  and $y_t^{[1]}=0$) } \\
&\neq q^{[2]}(\alpha_{t-1}^{[2]}-2)+2 \qquad  \qquad \text{     (since  $q^{[1]}<1$)}  \\
& = q_{t}^{[2]}\alpha_{t-1},  \qquad  \qquad \text{     (by equation \eqref{eq.42}) }
\end{align*}
Similarly, we can check that $$
q_{t+1}^{*[2]}\beta_{t}^{[2]} \neq q_t^{*[2]}\beta_{t-1}^{[2]}.
$$

In other words, Model 3 assumes that even for a non-claim period, its associated random effect continues to evolve.  While this assumption is acceptable, it is interesting to check whether alternative updating rules can be applied. In particular, is it possible to not update the distribution of the random effect $\theta_t^{[2]} $ if $y_t^{[1]}$ is not observed? The answer is affirmative. Indeed, this amounts to slightly change Model 3 and distinguish the updating rule, according to whether or not $y_t^{[1]}=0$:
\begin{itemize}
\item If this equality does not hold, then we use the same updating rule \eqref{eq.36}.
\item If instead $y_t^{[1]}=0$, then we keep updating the frequency random effect through the same rule $$\theta_t^{[1]} =\theta_{t-1}^{[1]}\frac{B_t^{[1]}}{q^{[1]}},$$
but keep the same value for the severity random effect:
$$
\theta_t^{[2]} =\theta_{t-1}^{[2]}.
$$
\end{itemize}

This variant of Model 3, which proposes a distinct treatment of the case $y^{[1]}_t=0$, has a nice interpretation in terms of two-part models.  In the recent insurance literature, the term ``two-part" refers to a model with two response variables frequency (count valued) and severity (continuously valued) per period.  Interestingly ,  ``two-part" models first appeared in the econometric literature (for applications such as the health care cost of different individuals), where it is used for single, nonnegative response variable, whose distribution has a point mass at zero. In other words, there, the two parts are the positive expense (in which case only the total expense is observed) and non-expense cases \citep{cragg1971some,mullahy1998much}. This type of zero/continuous two-part models was later introduced in credit risk [see e.g. \cite{tong2013zero}] to disentangle positive loss with zero loss in the case of default,  as well as in insurance \citep{frees2011predicting, frees2013actuarial, shi2018pair,yang2020nonparametric} to distinguish between claim and non-claim cases.   Later on, this terminology evolved in the insurance literature and is now also used for models analyzing separately the frequency and severity components of the claims, where frequency is count, instead of binary valued \citep{shi2020regression, oh2021copula}.

Thus the above variant can be interpreted as a ``three-part" model by mixing the ideas of the claim/non-claim, and frequency/severity two-part models. More precisely, on the one hand, frequency and severity variables are both observed, have their own random effects and regression equations; on the other hand, when there is no claim, we allow for a potentially different updating rule of the severity component, by explicitly acknowledging that in this case, the claim amounts are \textbf{not} observed and hence its associated random effects needs not necessarily be updated.  This terminology of three-part model is first introduced by \cite{shi2015dependent}, but their model is designed for cross-sectional data only.  Here, our three-part model is dynamic, and just as Model 3, it also allows for closed form predictive formulas.
\subsection{Combining EWMA and three-part variants}
Finally, the two variants proposed above can themselves be combined. That is, on the one hand, we keep the specification of the transition equations in the EWMA variant for periods when there is claim, but do not update the severity random effect when there is no claim. Such a combined variant will allow for $(\theta^{[2]}_t)$ to evolve as a martingale, while at the same time guarantee the existence of its predictive variance. These details are straightforward to prove and are omitted.
\subsection{Comparison with Model 3}
To summarize, we have proposed various variants for model 3. Even though all these models are tractable and reasonable, they have their own advantages and downsides. On the one hand, Model 3 is mathematically simpler, since the case $y^{[1]}_t=0$ is treated indifferently from its opposite case; on the other hand, the advantage of some its variant is that by letting the updating rule of the severity component depend on the sign of $y^{[1]}_t$, we have implicitly introduced a new channel of
 the \textit{dependence} between the frequency and severity parts, which is different from putting $y_t^{[1]}$ in the regression equation of $y_t^{[2]}$. This partly alleviates the  weakness of assuming $\theta_t^{[1]}$ and $\theta_t^{[2]}$ independent given the past $\mathcal{F}_{t-1}$ in Model 3. Indeed, we are not aware of any bivariate distribution which $i)$ has Gamma marginal distributions $ii)$ allows for non-trivial dependence between the two components $iii)$ is conjugate prior to the bivariate conditional distribution specified in equations \eqref{eq.38} and \eqref{eq.39}. This difficulty has also been acknowledged by \cite{abdallah2016sarmanov}, who show that in a bivariate frequency process, when the Sarmanov distribution is used to couple two random effect processes with Gamma marginal densities, the conjugacy property is lost and numerical approximation is necessary in order to obtain a tractable predictive mean.
\section{Empirical Analysis}

To assess the novelty of the proposed method, we perform an empirical analysis using a real dataset. For comparison, we introduce three benchmark models on top of the proposed model, which are special or limiting cases of Model \ref{mod.1} as follows:
\begin{itemize}
	\item Naive - No consideration is given to possible dependence between frequency and severity nor serial dependence among the claims of the same policyholder so that the expected compound loss is given as $\lambda_t^{[1]}\lambda_t^{*[2]}$. Note that this is equivalent to Model \ref{mod.1} where $q^{[1]}=q^{[2]}=1$, $\alpha_0^{[1]}=\beta_0^{[1]}=\alpha_0^{[2]}=\beta_0^{[2]}=\infty$, and $\eta=0$.
	\item DGLM - It only considers possible between frequency and severity but there is no serial dependence as in \citet{garrido2016generalized} so that the expected compound loss is given as $\lambda_t^{[1]}\lambda_t^{*[2]}\exp\left(\lambda_t^{[1]}(e^\eta-1)+\eta\right)$. Note that this is equivalent to Model \ref{mod.1} where $q^{[1]}=q^{[2]}=1$ and $\alpha_0^{[1]}=\beta_0^{[1]}=\alpha_0^{[2]}=\beta_0^{[2]}=\infty$.
	\item Static - Model \ref{mod.1} where $q^{[1]}=q^{[2]}$ are fixed as 1 and $\alpha_0^{[1]}, \alpha_0^{[2]}$ are estimated by maximizing the joint likelihood. Note that this is equivalent to the static credibility premium in \citet{jeong2020predictive} and \citet{jeong2020bregman}. In order to assure $\E{\theta_0^{[1]}}=\E{\theta_0^{[2]}}=1$, $\beta_0^{[1]}$ and $\beta_0^{[2]}$ are set as $\alpha_0^{[1]}$ and $\alpha_0^{[2]}-1$, respectively.

	\item Proposed - Model \ref{mod.1} where the dependence parameters are estimated by maximizing the joint likelihood with respect to $q^{[1]},q^{[2]},\alpha_0^{[1]}$, and $\alpha_0^{[2]}$. Again, in order to assure $\E{\theta_0^{[1]}}=\E{\theta_0^{[2]}}=1$, $\beta_0^{[1]}$ and $\beta_0^{[2]}$ are set as $\alpha_0^{[1]}$ and $\alpha_0^{[2]}-1$, respectively.
\end{itemize}

To calibrate the models, we used LGPIF (Local Government Property Insurance Fund) data from the state
of Wisconsin that has been widely used in actuarial literature including but not limited to \citet{lee2019dependent} and  \citet{yang2020nonparametric}. The dataset contains policy characteristics and observed claim on multiple lines of insurance including building and contents, inland marine (IM), and so on. In our work, we focused on the observed IM claims that consists of 5,677 observations over 5 years (2006--2010) for 1,234 policyholders. Note that the observations from year 2011 are set aside for out-of-sample validation. For detailed information and description, see \citet{frees2016multivariate} while brief summary statistics of the observed covariates are given in Table \ref{tab:datasummary}.

\begin{table}[h!t!]
\begin{center}
\caption{Observable policy characteristics used as covariates} \label{tab:datasummary}
\resizebox{!}{3cm}{
\begin{tabular}{l|lrrr}
\hline \hline
Categorical & Description &  & \multicolumn{2}{c}{Proportions} \\
variables \\
\hline
TypeCity & Indicator for city entity:           & Y=1 & \multicolumn{2}{c}{14.00 \%} \\
TypeCounty & Indicator for county entity:       & Y=1 & \multicolumn{2}{c}{5.78 \%} \\
TypeMisc & Indicator for miscellaneous entity:  & Y=1 & \multicolumn{2}{c}{11.04 \%} \\
TypeSchool & Indicator for school entity:       & Y=1 & \multicolumn{2}{c}{28.17 \%} \\
TypeTown & Indicator for town entity:           & Y=1 & \multicolumn{2}{c}{17.28 \%} \\
TypeVillage & Indicator for village entity:     & Y=1 & \multicolumn{2}{c}{23.73 \%} \\
\hline
 Continuous & & Minimum & Mean & Maximum \\
 variables \\
\hline
CoverageIM  & Log coverage amount of IM claim in mm  &  0 & 0.85
            & 46.75\\
lnDeductIM  & Log deductible amount for IM claim     &  0 & 5.34
            & 9.21\\
\hline \hline
\end{tabular}}
\end{center}
\end{table}

These observable policy characteristics are incorporated in risk classification via the associated parameters so that $\lambda_t^{[1]} = \exp(\mathbf{x}_t \zeta^{[1]})$ and $\lambda_t^{*[2]} = \exp(\mathbf{x}_t \zeta^{[2]})$.

For estimation of the parameters, we take a two-step approach so that the regressions parameters are estimated first due to the common mean structure, and the dependence structure is considered later. In other words, we fit a usual Poisson GLM to estimate $\zeta^{[1]}$ and Gamma GLM to estimate $\zeta^{[2]}$ and $\eta$. Such a sequential approach allows us to focus on the impact of dependence structure on the premium calculation. Table \ref{tab:fixed} summarizes the estimated values of $\zeta^{[1]}$, $\zeta^{[2]}$ , and $\eta$. Since $\eta=0$ is imposed in Naive model, there are two GLM's for the severity component where the second model does not preclude non-zero value of $\eta$. Note that negative value for estimated coefficient for ``Count'', $\eta$, implies there is negative dependence between the number of claim and average amount of claim in a year.



\begin{table}[!h]
\centering
\caption{Regression estimates of the fixed effects} \label{tab:fixed}
\begin{tabular}{lrrrrrr}
\toprule
\multicolumn{1}{c}{ } & \multicolumn{2}{c}{Frequency} & \multicolumn{4}{c}{Severity} \\
\cmidrule(l{3pt}r{3pt}){2-3} \cmidrule(l{3pt}r{3pt}){4-7}
\multicolumn{1}{c}{ } & \multicolumn{2}{c}{ } & \multicolumn{2}{c}{Independent} & \multicolumn{2}{c}{Dependent} \\
\cmidrule(l{3pt}r{3pt}){4-5} \cmidrule(l{3pt}r{3pt}){6-7}
  & Estimate & p-value & Estimate & p-value & Estimate & p-value\\
\midrule
(Intercept) & -4.2571 & 0.0000 & 9.4271 & 0.0000 & 10.1312 & 0.0000\\
TypeCity & 0.9673 & 0.0000 & 1.0040 & 0.0339 & 1.3007 & 0.0012\\
TypeCounty & 1.8747 & 0.0000 & 1.3164 & 0.0060 & 1.3152 & 0.0009\\
TypeMisc & -2.7453 & 0.0068 & -1.3007 & 0.6018 & -1.2977 & 0.5252\\
TypeSchool & -0.9174 & 0.0008 & -0.2046 & 0.7617 & -0.1191 & 0.8300\\
TypeTown & -0.4129 & 0.1254 & 0.2571 & 0.6981 & 0.3588 & 0.5089\\
CoverageIM & 0.0687 & 0.0000 & 0.0008 & 0.9712 & 0.0332 & 0.0746\\
lnDeductIM & 0.1363 & 0.0030 & -0.1336 & 0.4481 & -0.1744 & 0.2278\\
Count &  &  &  &  & -0.4538 & 0.0000\\
\hhline{=======}
\end{tabular}
\end{table}

After the models are calibrated, one can compare the prediction performance of each model by using out-of-validation set. We used root-mean-square error (RMSE) and mean absolute error (MAE) for quantifying discrepancies between the actual and predicted values via $L_2$ and $L_1$ norms, respectively. When we consider the credibility factors multiplied to the prior means, $\frac{\alpha_t^{[1]}}{\beta_t^{[1]}} \frac{\beta_t^{[2]}}{\alpha_t^{[2]}-1}$, they are capped at 250\% to mimic usual practice in posterior ratemaking.  We also compared the average of predictive values and the actual claim amounts for year 2011 for each model. As shown in Table \ref{tab:validation}, the proposed model is the best in terms of MAE, second best in terms of RMSE, and also closely matches the overall portfolio mean of the actual claims compared to the other benchmarks.



\begin{table}[h!]
\centering
\caption{Out-of-sample Validation Result} \label{tab:validation}
\begin{tabular}{c|ccccc}
\hline \hline
                       & Naive      & DGLM    & Static       & Proposed   & Actual   \\ \hline
RMSE               & 9272.96  & 6433.21  & 6340.25  & 6389.32  &       -       \\
MAE                 & 1345.41  & 1241.95  & 1101.32  & 1085.89  &       -       \\
Portfolio Mean   & 310.56    & 384.55   & 522.91    & 549.00    & 645.25    \\
\hline \hline
\end{tabular}
\end{table}





Finally, one can see that the ages of claims affect credibility factors of each model in different ways. Let us consider a hypothetical policyholder whose $\lambda_t^{[1]}=0.2$ and $\lambda_t^{[2]}=15000$ for $t=1,\ldots, 5$. For simplicity, we further assume that $\alpha_0^{[1]}=1$, $\alpha_0^{[2]}=3$, $\eta=0$, and $\psi=1.5$. Suppose that there was exactly one claim by the policyholder in one of the years 1 through 4. As shown in Table \ref{tab:hypo}, we observe decaying effects of past claims in the calculation of credibility factors under the dynamic credibility model, whereas the static model disregards such information.

\begin{table}[H]
\centering
\caption{Variation of credibility factors at year 5 under different models} \label{tab:hypo}
\begin{tabular}{ccccc}
\toprule
\multicolumn{1}{c}{ } & \multicolumn{2}{c}{Dynamic Credibility ($q^{[1]}=q^{[2]}=0.8$)} & \multicolumn{2}{c}{Static Credibility ($q^{[1]}=q^{[2]}=1$)} \\
\cmidrule(l{3pt}r{3pt}){2-3} \cmidrule(l{3pt}r{3pt}){4-5}
Claim year & Frequency & Severity & Frequency & Severity\\
\midrule
1 & 0.9216 & 1.0309 & 1.1111 & 1.0309\\
2 & 1.0496 & 1.0347 & 1.1111 & 1.0309\\
3 & 1.2096 & 1.0385 & 1.1111 & 1.0309\\
4 & 1.4096 & 1.0421 & 1.1111 & 1.0309\\
\hhline{=====}
\end{tabular}
\end{table}

\section{Conclusion}
In this paper we have introduced a new dynamic collective risk model to the insurance literature. The model comes with time-varying random effect processes,which allows to account for the seniority of claims in the prediction formula. This latter has many nice features. First of all it is available in closed form, thanks to a careful combination of ideas from the time series literature on univariate processes of counts and positive real numbers, respectively. Secondly, the predictive mean is a simple function of previous claim frequency/severity whose accounts for the seniority of the claim. Finally, we have illustrated the usefulness of our model using a longitudinal database.

\section*{Acknowledgements}

Jae Youn Ahn was supported by a National Research Foundation of Korea (NRF) grant funded by the Korean Government (2020R1F1A1A01061202). Yang Lu thanks NSERC through a discovery grant (RGPIN-2021-04144, DGECR-2021-00330). Himchan Jeong was supported by the Simon Fraser University New Faculty Start-up Grant (NFSG).

\bibliographystyle{apalike}
\bibliography{CTE_Bib_HIX}


\end{document}